\pdfoutput=1
\documentclass[a4paper,USenglish]{lipics-v2016}
 
\usepackage{cite}                 

\theoremstyle{plain}
\newtheorem{claim}[theorem]{Claim}

\renewcommand{\epsilon}{\varepsilon}

\DeclareMathOperator{\col}{\mathit{color}}

\newcommand{\old}[1]{{}}
\newcommand{\versionA}[1]{{}}

\newcommand{\eps}{\varepsilon}

\definecolor{darkblue}{rgb}{0,0,0.6}

\title{A PTAS for vertex guarding weakly-visible polygons --- An extended abstract}
\author[1]{Matthew J. Katz}
\affil[1]{Department of Computer Science, Ben-Gurion University of the Negev,
  Beer-Sheva 84105, Israel;\\
  \texttt{matya@cs.bgu.ac.il}.}

\begin{document}
\maketitle


\begin{abstract}
In this extended abstract, we present a PTAS for guarding the vertices of a weakly-visible polygon $P$ from a subset of its vertices, or in other words, a PTAS for computing a minimum dominating set of the visibility graph of the vertices of $P$. We then show how to obtain a PTAS for vertex guarding $P$'s boundary.
\end{abstract}

\section{Introduction}
\label{sec:intro}

Let $P$ be a polygon and let $e=(u,v)$ be one of its edges. We say that $P$ is \emph{weakly visible} from $e$, if every point in $P$ is visible from a point on $e$. Notice that if the angles at $u$ and at $v$ are both convex, then $P \setminus e$ is contained in one of the open half planes defined by the line containing $e$. In this paper, we fist consider an $n$-gon $P$, which is weakly visible from one of its edges $e=(u,v)$ where the angles at both $u$ and $v$ are convex. Without loss of generality, we assume that $e$ is contained in the $x$-axis and that $P \setminus e$ is contained in the open half plane above the $x$-axis. We study the problem of guarding the vertices of $P$ from a subset of its vertices. That is, we seek a minimum-cardinality subset $Q$ of the vertices of $P$, such that for each vertex $w$ of $P$ there exists a vertex in $Q$ that sees $w$. We present a PTAS for this problem, i.e., we present a polynomial-time algorithm that computes such a guarding set of size $O(1+\eps)\cdot {\mbox{\it OPT}}$, for any $\eps > 0$, where {\it OPT} is the size of a minimum-cardinality such guarding set. We then show how to remove the assumption that the angles at both $u$ and $v$ are convex. Finally, we show how to obtain a PTAS for vertex guarding $P$'s boundary.

Our PTAS is a standard local search algorithm. Its proof is based on the observation that the, so called, \emph{order claim}, which was originally stated for 1.5D terrains (see~\cite{Ben-MosheKM07}), also holds for weakly visible polygons. We then adapt the proof of Krohn et al.~\cite{KrohnGKV14}, who presented a PTAS for vertex guarding the vertices of a 1.5D terrain, to our setting. The proof of Krohn et al.~\cite{KrohnGKV14}, in turn, is based on the proof scheme of Mustafa and Ray~\cite{MustafaR09}, which is used to show that a local search algorithm is a PTAS (see also~\cite{ChanH12}).

{\bf Related results.}
The most relevant results are due to Bhattacharya et al.~\cite{BhattacharyaGR17}, who presented a 4-approximation algorithm for vertex guarding the vertices of a weakly-visible polygon and a 6-approximation algorithm for vertex guarding such a polygon.
Recently, by applying these results, Bhattacharya et al.~\cite{BhattacharyaGP17} managed to obtain the first constant approximation algorithm for vertex guarding a simple polygon. By an inapproximability result of Eidenbenz et al.~\cite{EidenbenzSW01}, this latter problem does not admit a PTAS, even if the goal is only to guard the polygon's boundary. Bhattacharya et al.~\cite{BhattacharyaGR17} also showed that vertex guarding a weakly-visible polygon with holes does not admit a polynomial-time approximation algorithm with approximation ratio better than $((1-\eps)/12) \ln n$, for any $\eps > 0$. As mentioned above, our main result builds on the result of Krohn et al.~\cite{KrohnGKV14}, who presented a PTAS for vertex guarding the vertices of a 1.5D terrain.

\section{Algorithm}

Let $V$ denote the set of vertices of $P$.
Given $\eps > 0$, set $k = \frac{\alpha}{\eps^2}$, for an appropriate constant $\alpha > 0$.
\begin{enumerate}
\item
$Q \leftarrow V$.
\item
Determine whether there exist subsets $S \subseteq Q$ of size at most $k$ and $S' \subseteq (V \setminus Q)$ of size
at most $|S| - 1$, such that $(Q \setminus S) \cup S'$ guards $V$.
\item
If such $S$ and $S'$ exist, set $Q \leftarrow (Q \setminus S) \cup S'$, and go back to Step~2. Otherwise, return $Q$.
\end{enumerate}

As usual, the running time of the algorithm is $O(n^{O(1/\eps^2)})$.

\section{Analysis}

For two points $a$ and $b$ on $P$'s boundary, $a \ne b$, we say that $a$ precedes $b$ (or $b$ succeeds $a$) and write $a \prec b$ (or $b \succ a$), if when traversing $P$'s boundary clockwise from $u$, one reaches $a$ before $b$.  

We first observe that the following claim, which was formulated for 1.5D terrains and is known as the \emph{order claim}, also holds for weakly visible polygons.

\begin{figure}[h]
	\centering
	\includegraphics[width=0.6\textwidth]{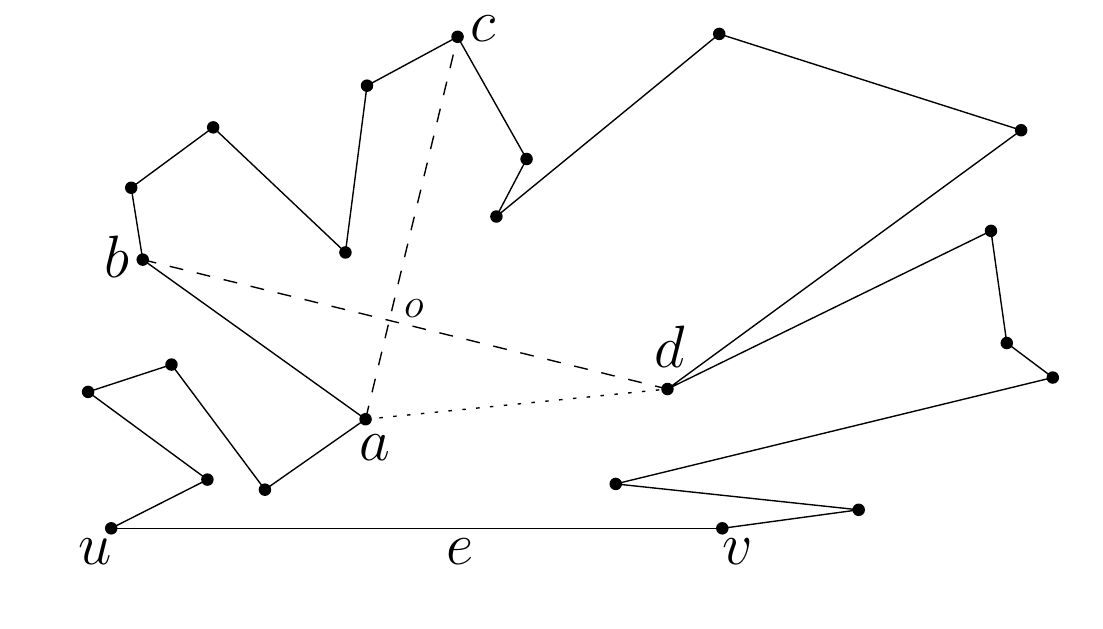}
	\caption{\label{fig:order_claim} A polygon weakly visible from $e=(u,v)$. The order claim: $a \prec b \prec c \prec d$, $a$ sees $c$, $b$ sees $d$ $\Longrightarrow$ $a$ sees $d$.}
\end{figure}

\begin{claim}[(Clockwise) order claim] Let $a, b, c, d$ be four vertices (or points on $P$'s boundary) such that $a \prec b \prec c \prec d$, and assume that $a$ sees $c$ and $b$ sees $d$. Then $a$ must also see $d$. 
\end{claim}

\begin{proof}
If $a$ does not see $d$, then either $a$ or $d$ is not visible from a point on $e$, see Figure~\ref{fig:order_claim}. Let $o$ denote the intersection point of $\overline{ac}$ and $\overline{bd}$. If the ray from $a$ in the direction of $d$ hits $P$'s boundary before reaching $d$, then $P$'s boundary enters and leaves the triangle $\Delta aod$ through the edge $ad$ without intersecting the edges $ao$ and $od$. If this happens before the boundary `reaches' $a$ (advancing clockwise from $u$), then $a$ cannot be seen from $e$, and if this happens before the boundary `reaches' $d$ (advancing counterclockwise from $v$), then $d$ cannot be seen from $e$. 
\end{proof}

Let $R$ (the red set) be a minimum-cardinality guarding set and let $B$ (the blue set) be the guarding set obtained by the algorithm above. We need to prove that $|B| \le (1+\eps)\cdot |R|$. We may assume that $R \cap B = \emptyset$; otherwise, we prove that $|B'| \le (1+\eps)\cdot |R'|$, where $R' = R \setminus B$ and $B' = B \setminus R$.  
We construct a bipartite graph $G = (R \cup B, E)$, and prove that (i) $G$ is planar and (ii) $G$ satisfies the \emph{locality condition}, that is, for any vertex $w$, there exist vertices $r \in R$ and $b \in B$, such that $r$ sees $w$, $b$ sees $w$, and $(r,b) \in E$. By the proof scheme of Mustafa and Ray~\cite{MustafaR09}, this implies that $|B| \le (1+\eps)\cdot |R|$.

\begin{figure}[h]
	\centering
	\includegraphics[width=0.7\textwidth]{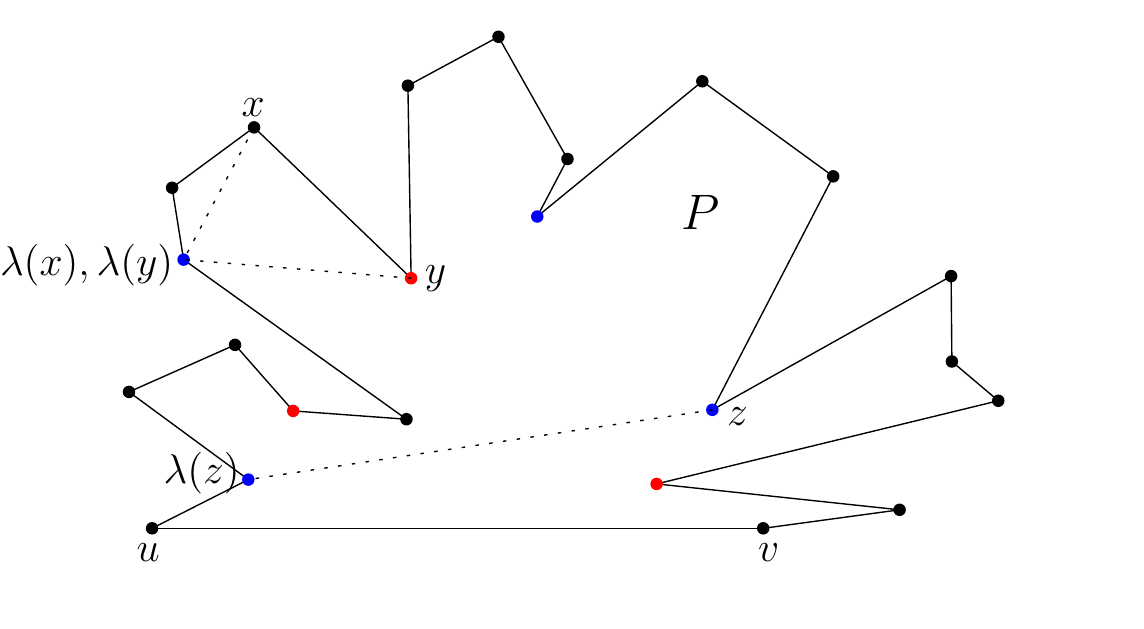}
	\caption{\label{fig:A1} The sets $R$ and $B$ in red and blue, and vertices $x, y, z$ of $P$ and their corresponding $\lambda$-vertices in $R \cup B$.}
\end{figure}

For a vertex $w$ of $P$, if there exists a vertex in $R \cup B$ that sees $w$ and precedes it, then let $\lambda(w)$ be the first such vertex (i.e., when traversing the boundary clockwise from $u$); see Figure~\ref{fig:A1}. Similarly, if there exists a vertex in $R \cup B$ that sees $w$ and succeeds it, then let $\rho(w)$ be the last such vertex. Notice that since $R \cap B = \emptyset$ at least one of the two exists. 

\noindent
{\bf Constructing $G$.}
Let $A_1 = \{\overline{\lambda(w)w} \ | \ w \mbox{ a vertex of } P \mbox{ for which } \lambda(w) \mbox{ is defined}\}$.

\begin{claim}
\label{cl:A1_non-crossing}
The segments in $A_1$ are non-crossing.
\end{claim}
\begin{proof}
Let $\overline{\lambda(x)x}$ and $\overline{\lambda(y)y}$ be two segments in $A_1$, such that $\lambda(x) \ne \lambda(y)$. Assume, w.l.o.g., that $\lambda(x) \prec \lambda(y)$. We first notice that it is impossible that $\lambda(x) \prec \lambda(y) \prec x \prec y$, since by the order claim, this would imply that $\lambda(x)$ sees $y$, which is impossible by the definition of $\lambda(y)$. Therefore, either (i) $\lambda(x) \prec \lambda(y) \prec y \prec x$, or (ii) $\lambda(x) \prec x \prec \lambda(y) \prec y$. But, clearly, in both these cases the two segments cannot cross each other, even if the polygon is not weakly visible.   
\end{proof}

For each vertex $x \in R \cup B$, do the following. If $\lambda(x)$ is defined and $\col(\lambda(x)) \ne \col(x)$, add the edge $(\lambda(x),x)$ to $E_1$. If there exists a segment $\overline{ab} \in A_1$, such that $a \prec x \prec b$, then let $\overline{\lambda(w)w} \in A_1$ be the sole such segment that can be reached from $x$ without existing $P$ and without intersecting any other segment in $A_1$ (except possibly at $x$). Now, if $\col(\lambda(w)) \ne \col(x)$, add the edge $(\lambda(w),x)$ to $E_1$.

\begin{figure}[h]
	\centering
	\includegraphics[width=0.9\textwidth]{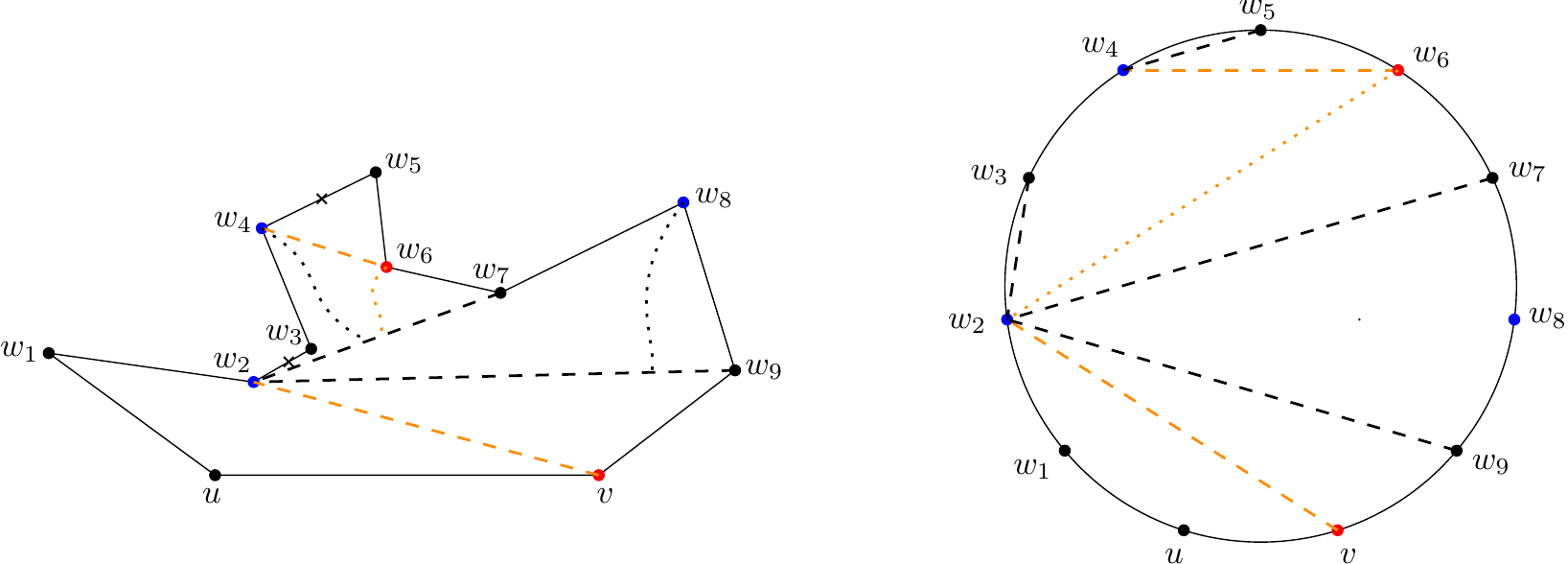}
	\caption{\label{fig:G1} Left: The dashed segments, including the two marked edges of $P$, are the segments in $A_1$. The segments in $A_1$ that are added to $E_1$ are drawn in orange. The dotted curves connect $w_4$, $w_6$, and $w_8$ to their corresponding segments in $A_1$, but only the one connecting $w_6$ induces an edge in $E_1$ and is therefore drawn in orange. Right: The embedding of $\overline{G}_1 = (V, A_1 \cup E_1)$.}
\end{figure}

\subsection{$G$ is planar}

We now prove that the bipartite graph $G_1 = (R \cup B, E_1)$ is planar, by describing an embedding of $G_1$ or, more precisely, of the graph $\overline{G}_1 = (V, A_1 \cup E_1)$. Let $C$ be the unit circle centered at the origin. We map the vertices in $V$ to equally-spaced points on $C$ and the edges in $A_1 \cup E_1$ to line segments between pairs of points, see Figure~\ref{fig:G1}.
We claim that the resulting set of line segments is non-crossing, i.e., we have obtained an embedding of $\overline{G}_1$ and therefore also of $G_1$. This follows from Claim~\ref{cl:A1_non-crossing} and by observing that the edges in $E_1 \setminus A_1$ can be partitioned into a collection of `fans', where each fan is associated with a segment $\overline{\lambda(w)w}$ of $A_1$ and lies to its left (when traversing the segment from $\lambda(w)$ to $w$).  

We now define the sets $A_2$ and $E_2$ by replacing $\lambda$ with $\rho$, that is,
$A_2 = \{\overline{w\rho(w)} \ | \ w \mbox{ a vertex of }\\ P \mbox{ for which } \rho(w) \mbox{ is defined}\}$ and $E_2$ is defined w.r.t. $A_2$. We then observe that the bipartite graph $G_2 = (R \cup B, E_2)$ is planar, by describing an embedding of $\overline{G}_2 = (V, A_2 \cup E_2)$. Moreover, we claim that the graph $G_1 \cup G_2$ is planar, since we can embed the graph $\overline{G}_1 \cup \overline{G}_2$ by drawing the edges of $\overline{G}_1$ inside $C$ and the edges of $\overline{G}_2$ outside $C$.

Finally, we define the set $E_3$ as follows. For each vertex $x \not\in R \cup B$, if both $\lambda(x)$ and $\rho(x)$ are defined and $\col(\lambda(x)) \ne \col(\rho(x))$, then add the edge $(\lambda(x), \rho(x))$ to $E_3$. The final graph $G = (R \cup B, E)$ where $E=E_1 \cup E_2 \cup E_3$ is planar, since $\overline{G}_1 \cup \overline{G}_2$ is planar and each edge $(\lambda(x), \rho(x)) \in E_3$ can be drawn as the union of the segments $\overline{\lambda(x)x} \in A_1$ and $\overline{x\rho(x)} \in A_2$.

\subsection{$G$ satisfies the locality condition}

We prove that the locality condition holds.
\begin{lemma}
For any vertex $w \in V$, there exist vertices $r \in R$ and $b \in B$, such that $r$ sees $w$, $b$ sees $w$, and $(r,b) \in E$.
\end{lemma} 
\begin{proof}
Let $x$ be a vertex of $P$. We distinguish between two cases:\\
{\bf $x \not\in R \cup B$:} If both $\lambda(x)$ and $\rho(x)$ are defined and $\col(\lambda(x)) \ne \col(\rho(x))$, then $(\lambda(x), \rho(x)) \in E_3$ and the condition holds. If both $\lambda(x)$ and $\rho(x)$ are defined but $\col(\lambda(x)) = \col(\rho(x))$, then there exists a vertex $w \in R \cup B$, such that $w$ sees $x$ and $\col(w) \neq \col(\lambda(x)), \col(\rho(x))$. Assume, w.l.o.g., that $\lambda(x) \prec w \prec x$ and let $z$ be the first such vertex (when traversing $P$'s boundary clockwise from $u$). Let $(\lambda(y),y)$ be the segment in $A_1$ associated with $z$. Then $\lambda(x) \preceq \lambda(y) \prec z \prec y \preceq x$. Notice the $\lambda(y)$ sees $x$, since if $y \ne x$, then by the order claim (applied to $\lambda(y), z, y, x$) $\lambda(y)$ sees $x$, and if $y = x$, then $\lambda(y) = \lambda(x)$ so $\lambda(y)$ sees $x$. Now, since $z$ is the ``first such vertex'',  $\col(z) \ne \col(\lambda(y))$, so the edge $(\lambda(y),z) \in E_1$ and the condition holds. If only $\lambda(x)$ is defined, then we proceed as above. 

\noindent
{\bf $x \in R \cup B$:} If $\lambda(x)$ is defined and $\col(x) \ne \col(\lambda(x))$, then $(\lambda(x),x) \in E_1$ and the condition holds. Similarly, if $\rho(x)$ is defined and $\col(x) \ne \col(\rho(x))$, then $(x,\rho(x)) \in E_2$ and the condition holds. Otherwise, we conclude w.l.o.g. that there exists a vertex $w \in R \cup B$, such that $w$ sees $x$ and $\lambda(x) \prec w \prec x$ and $\col(w) \neq \col(\lambda(x))$. Let $z$ be the first such vertex (when traversing $P$'s boundary clockwise from $u$), and proceed exactly as in the previous case.
 
\end{proof}

\section{Extensions}

\begin{figure}[h]
	\centering
	\includegraphics[width=0.6\textwidth]{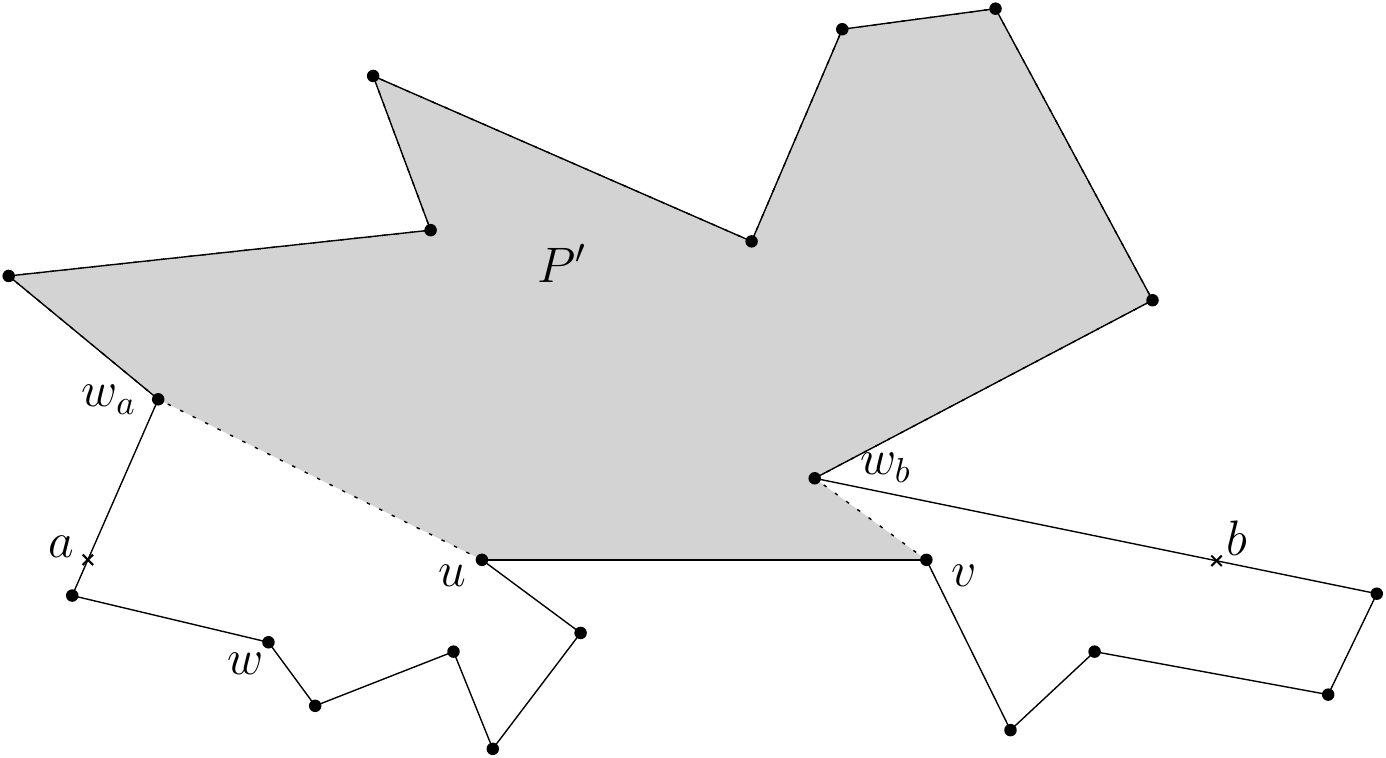}
	\caption{\label{fig:wv_poly} Removing the convexity assumption.}
\end{figure}

{\bf Removing the convexity assumption.}
We show how to remove the assumption that the angles at $u$ and at $v$ are convex. Assume, e.g., that the angle at $u$ is concave, and let $a$ be the first point on $P$'s boundary (moving clockwise from $u$) that lies on the $x$-axis; see Figure~\ref{fig:wv_poly}. Then, every point in the open portion of the boundary between $u$ and $a$ is visible from $u$ and is not visible from any other point on the edge $e=(u,v)$. Moreover, by the order claim, for any vertex $w$ in this portion of $P$'s boundary, if $w$ sees some point on $P$'s boundary, then so does $u$. Therefore, we may assume that an optimal guarding set does not include a vertex from this portion. Now, let $w_a$ be the first vertex following $a$. We place a guard at $u$ and replace the portion of $P$'s boundary between $u$ and $w_a$ by the edge $(u,w_a)$. Similarly, if the angle at $v$ is concave, we define the point $b$ and the vertex $w_b$ (by moving counterclockwise from $v$), place a guard at $v$, and replace the portion of $P$'s boundary between $v$ and $w_b$ by the edge $(v,w_b)$. Finally, we apply our local search algorithm to the resulting polygon $P'$, after adjusting $k$ so that together with $u$ and $v$ we still get a $(1 + \eps)$-approximation of an optimal guarding set for $P$.

{\bf Guarding the polygon's boundary from its vertices.}
In this paragraph we continue to assume that the angles at $u$ and $v$ are convex. 
We have described a PTAS for vertex guarding the vertices of $P$, however, with minor modifications, one can obtain a PTAS for vertex guarding a polynomial-size set $W$ of points on $P$'s boundary. To obtain a PTAS for vertex guarding the polygon's boundary, we generate a polynomial-size set of \emph{witness} points $W$ on $P$'s boundary, such that any subset of vertices that guards $W$, guards the entire boundary. This is done using ideas similar to those used in Friedrichs et al.~\cite{FriedrichsHK016}, who did it for 1.5D terrains.

{\bf Concluding remarks.}
It would be interesting to find other families of polygons for which a PTAS exists for vertex guarding the polygon's set of vertices (or its boundary or its boundary plus interior). In particular, does there exist a PTAS for vertex guarding the vertices of a simple polygon? Finally, it would be interesting to examine whether our results can be used to improve the constants of approximation obtained by Bhattacharya et al.~\cite{BhattacharyaGR17} for vertex guarding a weakly-visible polygon and by Bhattacharya et al.~\cite{BhattacharyaGP17} for the three versions of vertex guarding a simple polygon.

{\bf Acknowledgment.}
\old{
We wish to thank Nandhana Duraisamy, Ramesh Kumar, Anil Maheshwari, and Subhas Nandy for pointing out an error in the previous version of this manuscript, specifically, for misinterpreting an observation of Bhattacharya et al.~\cite{BhattacharyaGR17}. We had concluded that vertex guarding the boundary of a weakly-visible polygon implies vertex guarding the entire polygon. It is not difficult to find examples where this statement is evidently wrong.
}
We wish to thank Nandhana Duraisamy, Ramesh Kumar, Anil Maheshwari, and Subhas Nandy for pointing out an error in Section~4 (Extensions) of the previous version of this manuscript.

\bibliographystyle{abbrv}
\bibliography{p}

\old{

}

\end{document}